\documentclass[12pt,reqno,tbtags]{amsart}

\date{7 October, 2013}
\title{VCG Auction Mechanism Cost Expectations and Variances}

\usepackage{latexsym,amsfonts,amssymb,graphicx,lscape}
\usepackage{xspace}
\usepackage[normalem]{ulem}
\usepackage{amsthm}
\usepackage{ifpdf}
\usepackage[alphabetic,abbrev]{amsrefs}
\usepackage{hyperref}
\hypersetup{pdfstartview=}

\numberwithin{equation}{section}

\author{Svante Janson}
\thanks{SJ partly supported by the Knut and Alice Wallenberg Foundation}
\address[Svante Janson]
{Department of Mathematics, Uppsala University, PO Box 480,
SE-751~06 Uppsala, Sweden}
\email{svante.janson@math.uu.se
  \quad http://www2.math.uu.se/{\tiny$\sim$}svante/}

\author[Gregory B. Sorkin]{Gregory B. Sorkin}
\address[Gregory B. Sorkin]{
Department of Management \\
London School of Economics and Political Science \\
Houghton Street \\
London WC2A 2AE
}
\email{g.b.sorkin@lse.ac.uk}

\thanks{This paper was largely written during visits to 
Institut Mittag-Leffler, Sweden, 
and 
Mathematisches Forschungsinstitut Oberwolfach, Germany.}

\newtheorem{theorem}{Theorem}[section]
\newtheorem{lemma}[theorem]{Lemma}

\newtheorem{corollary}[theorem]{Corollary}
\theoremstyle{definition}
\newtheorem{example}[theorem]{Example}
\newtheorem*{example*}{Example}
\newtheorem{remark}[theorem]{Remark}

\newcommand{\set}[1]{\left\{#1\right\}}

\def\b{\beta}
\renewcommand{\Re}{\mathbb{R}}

\newcommand{\E}{\mathbb{E}}
\newcommand{\VCG}{{\operatorname{VCG}}}
\newcommand{\cstar}{c^*}
\newcommand{\cVCG}{c^\VCG}
\newcommand{\cideal}{c} 
\newcommand{\Sstar}{S^*}

\renewcommand{\SS}{{\mathbf S}}

\newcommand{\Pra}[1]{{\Pr\nolimits_a \left( {#1} \right) }}
\newcommand{\one}[1]{{\mathbf 1} \! \left[ {#1} \right] }

\newcommand{\Replus}{\Re_{\ge 0}}

\newcommand{\ignore}[1]{\relax}
\newcommand{\Iz}[1]{I_{#1}^0}
\newcommand{\Iinf}[1]{I_{#1}^\infty}
\newcommand{\IFE}{I_{F\setminus E,E}^{0,\infty}}
\newcommand{\Ixx}[2]{I_{#1,#2}^{0,\infty}}
\newcommand{\Minf}[1]{M_{#1}^\infty}
\newcommand{\cond}{ \mid } 

\newcommand{\vcgt}[2]{ \cVCG(#2; #1) }
\newcommand{\cinc}{{c^+}}

\newcommand{\rk}{r}

\newcommand\ga{\alpha}

\newcommand\gD{\Delta}

\newcommand\gl{\lambda}

\newcommand\gs{\sigma}
\newcommand\gss{\sigma^2}
\newcommand\eps{\varepsilon}

\newcommand\iid{i.i.d.\spacefactor=1000}     
\newcommand\ie{i.e.\spacefactor=1000}     
\newcommand{\as}{a.s.\spacefactor=1000}     
\newcommand{\cf}{c.f.\spacefactor=1000}     

\newcommand\bigpar[1]{\bigl(#1\bigr)}
\newcommand\Bigpar[1]{\Bigl(#1\Bigr)}

\newcommand\Bigsqpar[1]{\Bigl[#1\Bigr]}

\newcommand{\refT}[1]{Theorem~\ref{#1}}
\newcommand{\refC}[1]{Corollary~\ref{#1}}
\newcommand{\refL}[1]{Lemma~\ref{#1}}
\newcommand{\refR}[1]{Remark~\ref{#1}}
\newcommand{\refS}[1]{Section~\ref{#1}}

\newcommand{\refE}[1]{Example~\ref{#1}}

\newcommand\nota{{A\setminus a}}

\newcommand\notax{\nota}

\newcommand\mm[1]{\mathcal M(#1)}
\newcommand\mmx{\mathcal M}
\newcommand\qm[1]{[\mmx,\mmx]_{#1}}
\newcommand\Var{\operatorname{Var}}
\newcommand\Cov{\operatorname{Cov}}
\newcommand\Exp{\operatorname{Exp}}
\newcommand\cF{\mathcal F}
\newcommand\cG{\mathcal G}
\newcommand\oi{[0,1]}
\newcommand\intoi{\int_0^1}

\newcommand\kk{\kappa}
\newcommand\subb[1]{_{(#1)}}
\newcommand\eqd{\overset{\mathrm{d}}{=}}

\newcommand\xfrac[2]{#1/#2}

\newcommand{\tend}{\longrightarrow}
\newcommand\dto{\overset{\mathrm{d}}{\tend}}

\newcommand\ntoo{\ensuremath{{n\to\infty}}}
\newcommand{\cpriv}{c^{\operatorname{priv}}}
\newcommand{\ta}{\cVCG(a;I)}

\newcommand{\BB}{B}

\newcommand{\weight}{cost\xspace}
\newcommand{\vv}{c^{\operatorname{Xvcg}}}
\newcommand{\Varb}[1]{\Var\bigpar{#1}}
\newcommand{\Eb}[1]{\E\bigpar{#1}}

\begin{document}

\begin{abstract}  
We consider Vickrey--Clarke--Groves (VCG) auctions 
for a very general combinatorial structure,
in an average-case setting
where item costs are independent, identically
distributed uniform random variables. 
We prove
that the expected VCG cost is at least double 
the expected nominal cost, 
and exactly double when the desired structure is a basis of a bridgeless
matroid. 
In the matroid case we further show that, conditioned upon the VCG cost,
the expectation of the nominal cost is exactly half the VCG cost,
and we show several results on 
variances and covariances among the nominal cost, the VCG cost,
and related quantities.
As an application, we find the asymptotic variance of the VCG cost of the
minimum spanning tree in a complete graph with random edge costs.
\end{abstract}

\maketitle

\section{Introduction and outline} \label{intro}

We begin with a motivating example.  
We want to buy a specific structure in a graph, such
as a spanning tree.  
Each edge has a separate owner,
who has a minimum ``true'' cost for which he would sell the edge;
this cost is private, known only to the owner.
One obvious ``mechanism'' for choosing a spanning tree to buy
is to ask each owner for the price of his edge
(for example by a simple sealed-bid auction), and
then choose the minimum (cheapest) spanning tree (MST) 
and pay the owners of the chosen edges
the prices that they have demanded. The problem with this and many other
mechanisms is that the owners have an incentive to lie about their true
costs:
an owner that claims a higher price will get more money, unless the price
becomes so high that another spanning tree is chosen.

The  Vickrey--Clarke--Groves (VCG) auction 
\cites{vickrey,clarke,Groves}
is designed to
avoid this problem. By using a clever mechanism defined 
in \refS{SSVCG} below,
which typically will pay the owners of the chosen edges more than their
claimed costs, an owner maximizes his profit by claiming his true cost of
his edge
(assuming that the owners act without collusion).
Thus the selected 
structure will really be the cheapest one, and
therefore the most efficient from society's point of view (in terms of
using the least resources), although the price we pay for it is higher than
the true cost.
The main purpose of the present paper is to study this overpayment.

Examples are known for which the VCG mechanism 
results
in arbitrarily large overpayment;
see for example \cite{AT07}. 
A typical situation may be better represented by an average-case setting.
Previous work has shown that 
with independent, identically distributed uniform item costs,
the expectation of the VCG auction price $\cVCG$ 
is exactly 2 times that of the nominal cost $\cstar$
in a procurement auction for a minimum spanning tree, 
and larger by a factor asymptotically approaching 2 in auctions for 
a perfect matching in a complete bipartite graph,
and a path in a complete graph \cite{CFMS}.

A unifying explanation for these results was given by \cite{AHW},
see \refS{Spf1}.
Using their argument, we show
in \refT{Th1} that for any VCG auction with independent,
uniformly distributed item costs, 
$\E(\cVCG) \geq 2 \E(\cstar)$,
and in \refT{Th2} that in a matroid setting 
(including the minimum spanning tree example) 
$\E(\cVCG) = 2 \E(\cstar)$.
In the matroid case, we also show  further results.
In particular, 
a matroidal development of the same ideas leads to \refT{Thcond},
the much stronger result that
$\E\bigpar{\cstar\mid\cVCG} = \tfrac12\cVCG$ for the conditional expectation.
This implies that $\Cov(\cstar,\cVCG) =\frac12\Var(\cVCG)$,
and \refT{Thvar} gives a variety of other variance relations.

\subsubsection*{Outline}
\refS{SSVCG} introduces VCG auctions, related notation, 
and the basic properties that we will draw on.
\refS{Smain} states the theorems just described,
along with some related results.
Proofs are given in Sections \ref{Spf1}--\ref{Svar}.
Here we use the ideas of \cite{AHW} and our extension of them; 
we also use a very different type of matroidal argument,
combined with a martingale method,
to give a second proof of \refT{Th2} and to prove parts of the
variance results of \refT{Thvar}.

In \refS{Snonuniform} we give some results for
non-uniform cost distributions.
Theorems \ref{Tmono1} and \ref{Tmono2} extend
Theorems \ref{Th1} and \ref{Thcond} 
to any distributions with decreasing density functions.
In many probabilistic problems, exponential distributions yield simpler
results than uniform distributions, but  
this is not the case here:
Corollary~\ref{Cmono} and Theorem~\ref{Th1a}
give particular results for the exponential distribution
(which in this setting does not lead to equalities)
and to a Beta distribution $B(\ga,1)$
(which does).

We give some examples illustrating our theorems  in 
\refS{SEx}. In particular, \refE{EMST2} gives asymptotic results on the
expectation and variance of the VCG cost of a MST in $K_n$ with 
\iid{} $U(0,1)$ edge costs: 
the expectation is $\sim 2\zeta(3)$ as shown by \cite{CFMS}, and
the variance is $\sim (24\zeta(4)-18\zeta(3))/n$.
The other examples are much simpler, and serve partly as
counterexamples showing some ways in which our results cannot be improved.

\section{VCG mechanism, payment, and threshold}\label{SSVCG}

In order to define the VCG mechanism, we  introduce some
terminology. We consider, as in the example above, the case of a single
buyer and multiple independent sellers of different items.
(The buyer may also act as the auctioneer, 
or there may be a separate auctioneer, 
but in any case the two roles are equivalent for our purposes, and
we will speak of the buyer.)

A \emph{procurement} or \emph{reverse} auction instance $I=(A,\SS,c)$
consists of a set $A$ of \emph{items},
a \emph{collection $\SS$ of desired structures}, $\SS \subseteq 2^A$,
and a \emph{cost function} $c \colon A \rightarrow \Replus$
giving the {cost} $c(a)$ of each item $a$.
The buyer's goal is to obtain any structure $S \in \SS$,
paying as little as possible.
For any $S \in \SS$,
the (nominal) cost of structure $S$ is 
\begin{align}
\cideal(S) := 
 \sum_{a \in S} c(a)
  = \sum_{a \in A} c(a) \one{a \in S} . \label{idealcost} 
\end{align}
We use the phrase \emph{nominal cost}
to distinguish it from
the \emph{VCG cost} defined below.  
Let $\cstar$ be the minimum cost, 
\begin{align}
 \cstar = \cstar(I) := \min_{S \in \SS} \cideal(S),  \label{idealcostX}
\end{align}
and call any $S \in \SS$ achieving this minimum a \emph{minimum structure}.
(Explicitly showing $\cstar$ as a function of the instance $I$ will 
sometimes be important, for example in \eqref{defcstar}.)

The  VCG mechanism purchases 
a minimum structure $\Sstar$, just as a simple auction would.
(If there are several minimum structures, 
it chooses one arbitrarily. 
In the random setting that we will 
consider in our main results, the minimum structure is unique with
probability~1, so there is no issue,
but we will be precise about this throughout.)
We say that the items $a \in \Sstar$ are \emph{selected}.
No payment is made for any item not selected.
For each selected item, the owner is paid an amount 
at least as large as the nominal cost, and typically larger.
This payment, the \emph{VCG cost},
is defined as follows.

Given any instance $I$, for each item $a \in A$
define an instance $\Iz a$ to be identical to $I$ except that $c(a)=0$.
Likewise, let $\Iinf a$ be identical to $I$ except that $c(a)=\infty$.
Define the \emph{incentive payment} for item $a$ as
\begin{align}
 \cinc(a) := \cstar(\Iinf a) - \cstar(I) 
\ge0,
\label{defcstar}
\end{align}
and the \emph{VCG threshold} for item $a$ as
\begin{align}\label{nota}
\vcgt{I}{a}
:=
 \cstar(\Iinf a) - \cstar(\Iz a).
\end{align}
(The names will soon be justified by \eqref{split-VCG} and 
\eqref{ifin}--\eqref{ifnot}.)
The payment for item $a$ under the VCG mechanism is given by:
\begin{equation}
\text{VCG payment to $a$} := 
\begin{cases}
    0 & \mbox{if $a \notin \Sstar$}\\
    \vcgt I a = c(a)\!+\!\cinc(a) & \mbox{if $a \in \Sstar$,} \label{split-VCG}
\end{cases}
\end{equation}
with the implicit equality justified by \eqref{vcgc+}.
Note  that by \eqref{nota}, the VCG threshold for $a$,
and thus the VCG payment if $a$ is selected,
depends only on the costs of other items, not on $c(a)$.

Let $\Sstar$ be a minimum structure. 
If $a \in \Sstar$ then
$\cstar(I)-\cstar(\Iz a)=c(a)$, 
because if $a$ is already part of the minimum structure $\Sstar$,
then decreasing its cost to 0 means $\Sstar$ is still a minimum structure,
and the cost of $\Sstar$ decreases by $c(a)$.
Then, by \eqref{nota} and \eqref{defcstar},
\begin{align} 
\ta
&=
 \cstar(\Iinf a) - \cstar(\Iz a)
\notag \\&=
 [\cstar(I) - \cstar(\Iz a)] + [\cstar(\Iinf a) - \cstar(I)]
\notag \\ &=
c(a)+\cinc(a)  \label{vcgc+}
\\ & \geq
c(a) \label{VCGpayment}
 .
\end{align}
If $a\notin\Sstar$, then 
$\cstar(\Iinf a)= \cstar(I)$ 
so $\cinc(a)=0$;
moreover, $\cstar(I)\le \cstar(\Iz a)+c(a)$, 
and hence
\begin{equation}  \label{cvcgnot}
\vcgt{I}{a}
= \cstar(\Iinf a)- \cstar(\Iz a) 
= \cstar(I)- \cstar(\Iz a) \le c(a).
\end{equation}

Recapitulating, if $\Sstar$ is a minimum structure,
\begin{align}
a \in \Sstar 
\implies 
\cinc(a) \geq 0  
& \text{ and }
\vcgt{I}{a} =c(a)+\cinc(a) \geq c(a) ,
\label{ifin}
\\
a \notin \Sstar 
\implies 
\cinc(a) =0  
& \text{ and }
\vcgt{I}{a} \leq c(a).
\label{ifnot}
\end{align}

If $\Sstar$ is a unique minimum structure, we can say a bit more. 
In this case, if $a\in\Sstar$, then $\cstar(\Iinf a) >\cstar(I)$ so
$\cinc(a)>0$
and  there is strict inequality in \eqref{VCGpayment}.
And if $a\notin\Sstar$, then 
$\cstar(\Iz a)+c(a)>c(\Sstar)=\cstar(I)$, 
and thus there is strict inequality
in \eqref{cvcgnot}.
That is, if $\Sstar$ is unique then we cannot have equality in 
\eqref{ifin} nor in \eqref{ifnot}, so that
\begin{align}
\begin{aligned}
a \in \Sstar 
\iff 
\cinc(a) > 0  
& \iff
\vcgt{I}{a} 
\ge c(a) 
\\ & \iff
\vcgt{I}{a} > c(a) .
\end{aligned}
\label{useconditions}
\end{align}

If follows that any minimum structure $\Sstar$ 
contains every item $a$ with $c(a)<\ta$,
and no item with $c(a)>\ta$;
this justifies calling $\ta$ the VCG threshold.
Furthermore, 
$\Sstar$ is unique if and only if there are no
items with $c(a)=\ta$. 
(It is easily seen that 
if there is any such item then there is at least one
minimum $\Sstar$ containing $a$ and another without $a$.
Recall that 
in our continuous random setting, minimum structures will be unique with
probability~1.)

As said above, 
the VCG payment for a selected item $a$
depends only on the costs of other items, not on $c(a)$.
This  leads to the \emph{truthfulness} property of the VCG mechanism
stated earlier: 
it is always optimal for the owner of item $a$ to set its declared price,
$c(a)$, equal to the item's true value to him, call it $\cpriv(a)$.
If $\cpriv(a) < \ta$
then the owner will be happy to receive payment of $\ta$
and therefore should declare some price $c(a)<\ta$.
If on the other hand 
$\cpriv(a) > \ta$,
then the owner would rather not sell,
and therefore should declare some price $c(a)>\ta$.
In either case, 
setting $c(a)=\cpriv(a)$ will work.
Since the owner does not know $\ta$, declaring a cost 
$c(a) \neq \cpriv(a)$ merely risks 
the loss of a profitable sale 
(if $\cpriv(a) < \ta < c(a)$)
or a disadvantageous sale
(if $\cpriv(a) > \ta > c(a)$).

Let us now look at the total amount paid by the
VCG mechanism, which is the sum over all selected items, 
items $a \in \Sstar$, of the payment for $a$ given by \eqref{split-VCG}.
Let us start with the incentive payments $\cinc(a)$.
We use the common notation that $x_+ :=x$ 
if $x>0$, and 0 otherwise.
Define the \emph{VCG overpayment} 
for any minimum structure $\Sstar$
by 
\begin{multline}
\cinc(\Sstar)
:= \sum_{a\in\Sstar}\cinc(a)
\\
=\sum_{a \in A}\cinc(a) 
=   \sum_{a \in A} \bigpar{\vcgt I a-c(a)}_+
=: \cinc(I)
\label{vcgtot}
.  
\end{multline}
The second 
inequality relies on \eqref{ifnot},
and the third on \eqref{ifnot} and \eqref{ifin}.
The final quantity $\cinc(I)$ is well-defined 
because the middle expressions do not depend on $\Sstar$.
That is, the VCG overpayment is a function of the instance: 
it is the same for any minimum structure $\Sstar$.
Similarly, define the \emph{total VCG cost} by
\begin{multline}
 \cVCG(\Sstar) 
:= \sum_{a \in\Sstar}  \vcgt I a
  = \sum_{a \in\Sstar}\bigpar{c(a)+\cinc(a)} 
 \\
  = \cstar + \cinc(\Sstar)
  = \cstar + \cinc(I)
  =: \cVCG(I) ;
  \label{vcgtot0}
\end{multline}
again, $\cVCG(I)$ is well-defined because
neither $\cinc(\Sstar)$ nor $\cstar$ depends on $\Sstar$.

\section{Main results}\label{Smain}
Our main results concern the expectations of the costs $\cstar$ and $\cVCG$
when the item costs are independent random values.
In particular, we are interested in the case when the costs are uniformly
distributed on some intervals $[0,d_a]$;
in this case we have the following general result.

\begin{theorem} \label{Th1}
In the general VCG setting, if the costs $c(a)$ 
are independent uniform random variables, $c(a) \sim U(0,d_a)$,
then
\begin{align}\label{th1}
 \E(\cVCG) \geq 2 \E(\cstar) .
\end{align}
\end{theorem}
That is, the expected VCG cost is at least twice as large as the
nominal cost of the minimum structure:
the VCG mechanism is overpaying significantly.
While this is the natural way of thinking of the conclusion,
at a technical level our proofs, and results such as \refT{Thcond} below,
suggest that it is perhaps better thought of as
$\E(\cstar) \leq \tfrac12 \E(\cVCG)$.

We obtain further results in a \emph{matroid} setting where
$\SS$ is the family of \emph{bases} of a matroid $M$ on \emph{ground set} $A$.
We recall that there are many equivalent definitions of matroids
and many examples of different types; see e.g.\ \cites{Welsh,White}. 
One important example of the matroid case is
the motivating spanning-tree example at the beginning of the paper,
which translates to matroid language as follows. 
(All the ideas in the paper can be appreciated
by thinking of just this example,
but working with matroids gives greater generality,
and highlights the restricted set of properties needed.)

\begin{example}\label{EMST}
$G$ is a connected graph.  $M$, 
the \emph{graph matroid} (or \emph{cycle matroid}) of $G$, 
has \emph{ground set} $A=E(G)$, the edges of $G$,
and its {bases} are the spanning trees of $G$.
A minimum-cost structure is then a minimum-cost basis,
namely a minimum-\weight spanning tree (MST).
\end{example}

Let $M$ be a matroid with ground set $A$.
For any subset $S \subseteq A$, let $\rk(S)$ denote the \emph{rank} of $S$.
An element $a \in S$ is called a \emph{bridge} in $S$ if
$\rk(S \setminus a) < \rk(S)$.
(When $M$ is a graph matroid,
this definition conforms with that of a bridge in a graph.)
Let
\begin{equation}
  \label{bridges}
\b(S) := |\set{a \in S \colon \rk(S \setminus a) < \rk(S)}| 
\end{equation}
be the number of bridges in $S$.
The bridges of $M$ are those of its ground set $A$,
and $M$ is bridgeless if there is no such element.
(We will use the terms ``element'' and ``item'' interchangeably,
favoring element in the matroid context, 
and item in the auction context.)

For bridgeless
matroids using a common uniform distribution of the item costs,
we have exact equality in \refT{Th1}.
(We have to exclude bridges, since a bridge belongs to every minimum
structure and thus its VCG payment is $\infty$, so the total VCG cost is
infinite for a matroid with bridges, \cf{} \eqref{m2} below.)

\begin{theorem} \label{Th2}
For a bridgeless
matroid with costs $c(a)$ \iid{} uniform random variables, 
$c(a) \sim U(0,1)$,
\begin{align*}
 \E(\cVCG) = 2 \E(\cstar) .
\end{align*}
\end{theorem}

\refT{Th2}'s special case of a MST in a complete graph
was proved by \cite{CFMS}.
That paper also showed 
that the expected VCG cost is asymptotically (but not exactly)
equal to twice the nominal cost
for a minimum-\weight path between a pair of vertices in a complete graph, 
and for a minimum-\weight assignment in a complete bipartite graph.
These are special cases of \refT{Th1} but not of \refT{Th2},
and the asymptotic equality is discussed further in \refR{Rcfms}.

Proofs of the theorems above are given in \refS{Spf1}.
\refS{Spf2} gives
a second proof of Theorem~\ref{Th2} using the following
formulas of independent interest
for the nominal and VCG costs in the matroid case.
For $t \geq 0$ let 
\begin{equation}
A(t) := \set{a \in A \colon c(a) \leq t}  
\end{equation}
be the set of items of cost at most $t$.
(Recall that costs are by definition non-negative.) 

\begin{lemma}\label{L1}
Let $M$ be a matroid with arbitrary (random or non-random) costs $c(a)$.
Then
\begin{align}
 \cstar(M) &= \int_0^\infty \big(\rk(A)-\rk(A(t))\big) \, dt  \label{m1}
 \\
 \cVCG(M) &= \cstar(M) + \int_0^\infty \b(A(t)) \, dt .  \label{m2}
\end{align}
\end{lemma}

\refT{Th2} concerns the overall expectations of $\cVCG$ and $\cstar$.
In \refS{Scond2} we use 
a matroidal extension of the ideas in the first proof of \refT{Th2} to
show the following result on the conditional expectation of
$\cstar$ given $\cVCG$; note that it is a considerable strengthening of
\refT{Th2} (which follows by taking the expectation).

\begin{theorem}
  \label{Thcond}
For a bridgeless
matroid with costs $c(a)$ \iid{} uniform random variables, 
$c(a) \sim U(0,1)$,
\begin{equation}\label{thcond}
  \E\bigpar{\cstar\mid\cVCG} = \tfrac12\cVCG.
\end{equation}
\end{theorem}

As an immediate corollary, we obtain the following formula for some mixed
moments of $\cstar$ and $\cVCG$.

\begin{corollary}\label{Cmom}
For a bridgeless
matroid with costs $c(a)$ \iid{} uniform random variables, 
$c(a) \sim U(0,1)$, and any real $m\ge0$,
\begin{align}\label{cmom}
\E\bigpar{\cstar(\cVCG)^m}=\tfrac12\E\bigpar{(\cVCG)^{m+1}}.
\end{align}  
\end{corollary}

\begin{proof}
  Multiply \eqref{thcond} by $(\cVCG)^m$ and take the expectation.
\end{proof}

These results do not generalize to higher powers of $\cstar$.
\refE{EU} shows that, given $\cVCG$,
the conditional distribution of $\cstar$ is not determined,
although by \refT{Thcond} its conditional expectation is.
The same example shows that 
$\E((\cstar)^2)$ and $\E((\cVCG)^2)$
are not in fixed proportion, i.e.,
\refT{Th2} does not generalize to the second moment; see also \refT{T22}.
Similarly, while the power of $\cVCG$ in \refC{Cmom} is arbitrary, 
it can be seen from the example that there is no similar
formula involving the second or higher power of $\cstar$.
Nevertheless, there are some general relations between the variances and
covariances of $\cstar$ and $\cVCG$, and we prove the following in 
\refS{Svar}.
The proof will use \refT{T22}
(itself proved with the first of our two methods for proving \refT{Th2})
and \refL{L2} (proved with the second method).

\begin{theorem}
  \label{Thvar}
For a bridgeless
matroid with costs $c(a)$ \iid{} uniform random variables, 
$c(a) \sim U(0,1)$,
\begin{align}
\Cov(\cstar,\cVCG)&=\frac12\Var(\cVCG). \label{cov}
\\
  \Var(\cVCG) &= 4\Var(\cstar) - \Var(\cVCG-2\cstar), \label{vvcg}
\\
\Var(\cstar) &=\intoi\intoi\Cov\bigpar{r(A(s)),r(A(t))}\,ds\,dt, \label{vcstar}
\\
\Var\bigpar{\cVCG-2\cstar}
&=
\int_0^1 \bigpar{\rk(A)-\E\rk(A(t))}\, 2t\, dt  \label{v1}
\\&=
\int_0^1 \E(\b(A(t))) \,t\, dt  \label{v2}.
\end{align}
\end{theorem}

The proof of \refT{Thcond} uses the following 
very general 
results of
independent interest; the proofs are given in \refS{Scond2}.
We extend the notations $\Iz a$ and $\Iinf a$ from \refS{SSVCG}.
Given an instance $I$ and two disjoint sets of items $F$ and $E$,
let $I^{0,\infty}_{F,E}$ be the instance $I$ modified so that all items in
$F$ have cost 0 and all items in $E$ have cost $\infty$; if $E$ or $F$ is
empty we write just $I^0_F$ and $I^\infty_E$, respectively.
In analogy with \eqref{nota}, for any independent set $F$
and any $a \in F$, define an \emph{extended VCG threshold} by
\begin{align}
  \vv(F,a) :=
  \cstar(I^{0,\infty}_{F\setminus\set a,a})
- \cstar(I^{0}_{F}).
 \label{vcgF}
\end{align}
Thus $\cVCG(a)=\vv(\set a,a)$, 
and $\vv(F,a)$ extends several key properties of the simple VCG threshold.
One is immediate from the definition:
for $a \in F$, $\vv(F,a)$ has value independent of the costs on $F$,
just as $\cVCG(a)$ is independent of the cost of $a$.
Other extended properties follow from the lemma below.

\begin{lemma}\label{Lmom5}
Consider a bridgeless 
matroid with ground set $A$ and arbitrary
costs $c(a)$, 
and let $\Sstar$ be a minimum basis.
Then
\begin{align}\label{lmom5}
F \subseteq \Sstar 
 & \implies (\forall a \in F) \colon
  \cVCG(a) = 
  \vv(F,a) .
\intertext{%
Furthermore,}
 F \subseteq \Sstar
 & \implies
(\forall a \in F) \colon 
 c(a) \leq \vv(F,a) ,
 \label{obs1y}
\\
 F \subseteq \Sstar
 & \impliedby
(\forall a \in F) \colon 
 c(a) < \vv(F,a) ,
 \label{obs1z}
\intertext{%
and if all $2^{|A|}$ sums of costs over different sets of elements
are distinct,
}
 F \subseteq \Sstar
 & \iff 
(\forall a \in F) \colon 
 c(a) \leq \vv(F,a)
 \label{obs1x}
.
\end{align}
\end{lemma}

By \eqref{lmom5}, if $a \in F \subseteq \Sstar$, 
the extended VCG threshold is equal to the simple one. 
The key properties of the simple VCG threshold 
are \eqref{ifin}--\eqref{useconditions};
\eqref{obs1y}--\eqref{obs1x} give analogous properties for the extended VCG
threshold. 

\refL{Lmom5} shows that for $F\subseteq\Sstar$, 
the VCG payments $\cVCG(a)$ for the items in $F$ depend only on the
costs of the items \emph{not} in $F$.
(For this conclusion, we get the strongest result if we choose $F=\Sstar$.)
This is a substantial (and non-obvious) extension of the observation after
\eqref{split-VCG} that the payment $\cVCG(a)$ does not depend on $c(a)$.

A possibly interesting 
implication is that
collusion 
among a set $F$ of owners 
will not affect their payments, as long as $F \subseteq \Sstar$:
price-fixing 
can only affect the payments if some of its participants
are not amongst the auction winners.

\begin{theorem} \label{tmom5a}
Consider a bridgeless 
matroid with ground set $A$ and
costs $c(a)$ independent continuous random variables,
and let $F\subseteq A$ be an independent set.
Condition on the costs of all items not in $F$
(which by \eqref{vcgF} 
determines $\vv(F,a)$ for all $a \in F$), and 
on the event $F\subseteq\Sstar$
(which by \eqref{lmom5} determines $\cVCG(a)=\vv(F,a)$ for all $a \in F$).
Then the
conditional distributions of $c(a)$, $a\in F$, are independent with
$c(a)$ having the same distribution as 
the conditioned random variable $\bigpar{c(a)\mid c(a)\le \cVCG(a)}$,
where we regard $\cVCG(a)=\vv(a)$ as a constant. 
\end{theorem}

The final statement means, more formally, that if the conditioning 
in the theorem yields a value $\cVCG(a)=\vv(a)=x$, then
the conditional distribution of $c(a)$ equals the 
distribution of $\bigpar{c(a)\mid c(a)\le x}$ (where we do not condition on
anything else).

Note that this theorem does not assume uniform distributions:
the item costs may follow any continuous distributions,
which need not be identical.
We give some applications to
non-uniform cost distributions in \refS{Snonuniform}.

\section{A simple proof of Theorems \ref{Th1} and \ref{Th2}}\label{Spf1}
In this section we give a proof of Theorem \ref{Th1} 
which also yields \refT{Th2}.
This argument was sketched in \cite{AHW};
since it has not appeared in print, we give it in full here,
with kind permission of the authors. 

\begin{proof}[Proof of \refT{Th1}]
Since the distribution of each $c(a)$ is continuous, almost surely
$c(a)\neq\vcgt I a$ and thus $\Sstar$ is unique 
and \eqref{useconditions} holds.
The total VCG cost is thus  by \eqref{vcgtot0}, \as,
\begin{align}
\cVCG
 &=
  \sum_{a \in \Sstar} \vcgt I a
=
  \sum_{a \in A} \vcgt I a \one{c(a) \leq \vcgt I a}
\end{align}
and
the expected total VCG cost is 
\begin{align}
 \E\bigpar{\cVCG}
 &=
  \sum_{a \in A}\E\Bigpar{ \vcgt I a \one{c(a) \leq \vcgt I a}} .
 \label{SumVCG}
\end{align}
By the method of conditional expectations,
the expected contribution for item $a$ is
\begin{multline}
 \E\bigpar{ \vcgt I a \one{c(a) \leq \vcgt I a} }
 \\=
 \E_{\nota}\big( \E_a \big( \vcgt I a \one{c(a) \leq \vcgt I a} 
   \cond {\notax} \big)\big) 
\label{termVCG0}
\end{multline}
where in the latter expression 
the inner (conditional) expectation $E_a$ 
means the  expectation over the cost distribution for item $a$
conditioned upon the values of the
costs for all items \emph{except} $a$
(with slight redundancy, we also indicate this as a conditioning on $\notax$),
and the outer expectation $E_\nota$ means 
the expectation over the distribution of the costs for all other items.
Taking advantage of the independence of the item costs
and recalling that $\vcgt Ia$ does not depend on $c(a)$,
looking at the inner expectation above we have
\begin{multline}
 \E_a  \big( \vcgt I a \one{c(a) \leq \vcgt I a} 
   \cond {\notax} \big) 
  \\=
   \vcgt I a  \; \Pra{c(a) \leq \vcgt I a \cond \notax} . \label{termVCG}
\end{multline}

Similarly, for the total nominal cost $\cstar$ we have almost surely
\begin{align}
 \cstar &= \sum_{a \in A} c(a) \one{c(a) \leq \vcgt I a}, \\
\E \cstar& = \sum_{a \in A}\E\Bigpar{ c(a) \one{c(a) \leq \vcgt I a}}, 
 \label{SumIdeal} 
\end{align}
where the expected contribution for item $a$ is
\begin{align}
  \E_{\nota} & \big( \E_a \big( c(a) \one{c(a) \leq \vcgt I a} \cond {\notax} \big)\big)
  .
\label{termideal0}
\end{align}
Taking advantage of the independence of the item costs,
and that $a$ has uniform distribution $U(0,d_a)$,
in the inner expectation in \eqref{termideal0} we have
\begin{align}
 \E_a & \big( c(a) \one{c(a) \leq \vcgt I a} \cond {\notax} \big) \notag
 \\& =
   \E_a \big( c(a) \cond c(a) \leq \vcgt I a \big) 
      \; \Pra{c(a) \leq \vcgt I a \cond \notax} 
   \notag
 \\& =
   \tfrac12 \min\set{d_a, \vcgt I a}  \; \Pra{c(a) \leq \vcgt I a \cond \notax}
   \label{termIdeal}
 .
\end{align}

  Immediately, each term given by \eqref{termIdeal} is at most 
half that of the same term in \eqref{termVCG},
and thus the expectation \eqref{termideal0} is at most half the expectation 
\eqref{termVCG0};
correspondingly the sum 
of the expectations 
in \eqref{SumIdeal} at most 
half that of \eqref{SumVCG}.
That is, the expected nominal cost is at most half the expected VCG cost.
\end{proof}

\begin{proof}[First proof of \refT{Th2}]
In this 
special case of \refT{Th1},
all $d_a=1$. Moreover, by the matroid basis exchange property,
since $a$ is not a bridge, if $S$ is any basis containing $a$, there exists
$b\neq a$ such that $S\setminus\set{a}\cup\set{b}$ is a basis, and it
follows from \eqref{nota} that
\begin{equation}
  \label{le1}
\vcgt Ia\le1. 
\end{equation}
Hence, 
the minimum in \eqref{termIdeal} is achieved by $\vcgt I a$,
each term in \eqref{termIdeal} is exactly half that
in \eqref{termVCG}, and thus the proof above of \refT{Th1} yields equality.
\end{proof}

\begin{remark} \label{cap1}
  The proof of \refT{Th2} 
shows that equality holds in \refT{Th1} if and only if 
$\vcgt Ia\le d_a$ a.s., for every $a\in A$.
This is not 
generally true outside the matroid setting,
even if all $d_a$ are equal; see \refE{Epath} and \cite{CFMS}. 
As pointed out by \cite{AHW}, if we know \emph{a priori}
that all item costs are bounded by 1, we could use a modified VCG auction
where the VCG payment for each item is capped to be at most 1; this is
equivalent to introducing fictitious ``shadow copies'' with costs 1 of all
items. 
For this modified auction mechanism,
the argument above shows that
equality holds in \eqref{th1} \cite{AHW}.
\end{remark}

\begin{remark} \label{Rcfms}
The MST result of \cite{CFMS} is the special case of \refT{Th2}
given by \refE{EMST}.
We will not rederive the results of \cite{CFMS} 
for shortest paths and minimum-cost matchings,
but they may be thought of as following from the reasoning of \refR{cap1}:
in these settings the VCG costs are not deterministically capped at 1,
but they are less than 1 asymptotically almost surely,
leading to the asymptotic equalities found in \cite{CFMS}.
\end{remark}

\section{A matroidal proof of \refT{Th2}}\label{Spf2}
In this section we give a second proof of Theorem~\ref{Th2}
using properties specific to matroids,
and introducing arguments useful later.
We begin by proving \refL{L1}.

\begin{proof}[Proof of \refL{L1}]
We regard item $a$ as arriving at time $c(a)$. 
Let $S\in \SS$ be a structure (here, a basis), and select the elements of $S$ as they
arrive. At time $t$, we thus have selected the set $S\cap A(t)$ of items.
Since the structure $S$ is a basis, $|S|=r(A)$ and also
$S\cap A(t)$ is an independent
subset of $A(t)$, whence $|S\cap A(t)|\le r(A(t))$. Consequently,
\begin{equation}\label{cseq}
  \begin{split}
c(S)
&=\sum_{a\in S}c(a)
=\sum_{a\in S} \int_0^\infty \one{t<c(a)}\, dt	
=\int_0^\infty \sum_{a\in S} \one{t<c(a)}\, dt	
\\&
=\int_0^\infty \Bigpar{|S|-\sum_{a\in S} \one{t\ge c(a)}}\, dt	
\\&
=\int_0^\infty \bigpar{|S|-|S\cap A(t)|}\, dt	
\ge\int_0^\infty \bigpar{r(A)-r(A(t))}\, dt	.
  \end{split}
\raisetag{1.2\baselineskip}
\end{equation}

Now, suppose $S$ to be the greedily chosen basis, that is,
the one produced by the following algorithm:
take the items in order of arrival (increasing cost),
breaking ties arbitrarily,
and add an item to $S$
if it is independent of all items previously seen
(thus independent of the items already in $S$).
By construction, $S$ is independent, and at each time $t\ge0$,
$|S\cap A(t)|=r(A(t))$, since every new item that increases the rank of
$A(t)$ is added to $S$. Letting $t\to\infty$ yields $|S|=r(A)$, and thus $S$ is
a basis, so $S\in\SS$. Furthermore, there is equality in \eqref{cseq} for
this basis $S$.
This establishes \eqref{m1} of \refL{L1}.
(It also shows that $S$ is a minimum structure $\Sstar$,
reproving the well-known fact that the greedy algorithm is optimal
for matroids.
Indeed this property characterizes matroids, 
a fact known as the Rado-Edmonds theorem; see \cites{Rado,Gale,Edmonds}.)

For the VCG payment, note that the incentive payment $\cinc(a)$ of an item
$a$ by the definition \eqref{defcstar} is $\cstar(\Minf a)-\cstar(M)$,
where $\Minf a$ is the matroid $M$ with the cost $c(a)$ changed to $\infty$.
This changes the set $A(t)$ to $A(t)\setminus a$, for every $t<\infty$, and thus
\eqref{m1} yields
\begin{equation}\label{kn}
  \begin{split}
\cinc(a)&=
 \cstar(\Minf a)
-\cstar(M)
\\&
= \int_0^\infty \Bigpar{\rk(A)-\rk(A(t)\setminus a)} \, dt 
-\int_0^\infty \Bigpar{\rk(A)-\rk(A(t))} \, dt
\\&
=\int_0^\infty \Bigpar{\rk(A(t))-\rk(A(t)\setminus a)} \, dt .
  \end{split}
\end{equation}
Summing over $a$ we obtain by \eqref{vcgtot} the total VCG overpayment
\begin{equation}\label{VCGover-M}
  \begin{split}
\cVCG-\cstar
=\sum_{a \in A}\cinc(a) 
=\int_0^\infty \sum_{a\in A}\bigpar{\rk(A(t))-\rk(A(t)\setminus a)} \, dt .	
  \end{split}
\end{equation}
By definition, $\rk(A(t))-\rk(A(t)\setminus a)=1$ if $a$ is a bridge in
$A(t)$, and 0 otherwise, and thus, recalling the definition \eqref{bridges},
\begin{equation}\label{bridges-t}
\sum_{a\in A} \bigpar{\rk(A(t))-\rk(A(t)\setminus a)} 
=\beta(A(t)),
\end{equation}
the number of bridges in $A(t)$.
Combining \eqref{VCGover-M} and \eqref{bridges-t} yields \eqref{m2}.
\end{proof}

\begin{remark}
  If we change the cost of $a$ to 0, the set $A(t)$ is changed to
  $A(t)\cup\set a$ for every $t>0$. Applying \eqref{m1} to $\Iinf a$ and
  $\Iz a$ yields, similarly to \eqref{kn}, a formula for the VCG threshold 
for an item:
\begin{equation}\label{qk}
  \begin{split}
\vcgt Ia&
=\int_0^\infty \Bigpar{\rk(A(t)\cup\set a)-\rk(A(t)\setminus a)} \, dt .
  \end{split}
\end{equation}
\end{remark}

\begin{remark}
\label{Rkk} 
In the case of 
a graph matroid defined by a connected graph $G$ as in \refE{EMST},
for any $F \subseteq A$, 
$r(F)=n-\kk(F)$, where $n$ is the number of vertices in $G$ and $\kk(F)$
is the number of components of the subgraph $F$ of $G$ with vertex
set $V(G)$ and edge set $F$. 
Since $G$ is connected, $r(A)=n-1$. 
Hence \eqref{m1} can be written
\begin{align}
 \cstar(M) &= \int_0^\infty \big(\kk(A(t))-1)\big) \, dt. \label{kk}
\end{align}
This formula has been used 
by \cite{Janson} and \cite{CFIJS}
to study the cost of the MST in
a complete graph $K_n$ with
random \iid{} edge costs. See further \refE{EMST2}.
\end{remark}

We turn to the case of random costs with the uniform distribution $U(0,1)$.

\begin{lemma}\label{Ldiff}
Where each element $a \in A$ has cost $c(a) \sim U(0,1)$ independently,
\begin{align*}
 \E(\b(A(t))) = t \frac d{dt} \E(\rk(A(t))),
\qquad 0<t<1.
\end{align*}
\end{lemma}

\begin{proof}
We argue informally. (See also \refS{Svar} below, specifically the second
proof of \eqref{v1}--\eqref{v2}.)
Run the process $A(t)$ backwards in time, from $t=1$ to $t=0$.
At time $t$, all elements in $A(t)$ have 
distribution $U(0,1)$ conditioned upon being less than $t$,
which is to say distribution $U(0,t)$.
It follows that in the (backward) time interval $(t,t-dt)$,
each element of $A(t)$ is lost with probability $(1/t) dt$.
In particular, this holds for each of the $\b(A(t))$ bridges,
whose deletion would reduce the rank.
Thus the expected decrease in $\rk(A(t))$ is $\b(A(t))t^{-1} dt$.
Going forward in time, this means
\begin{equation*}
\frac d{dt} \E(\rk(A(t)))
= t^{-1} \E(\b(A(t))).
\qedhere 
\end{equation*}
\end{proof}

\begin{proof}[Second proof of \refT{Th2}]
Since $M$ is bridgeless, $\b(A(t))=\b(A)=0$ for $t>1$, and thus
by \eqref{m2}, \refL{Ldiff}, integration by parts and \eqref{m1},
\begin{equation}\label{ip}
  \begin{split}
\E\cVCG-\E\cstar
&=\int_0^1\E\b(A(t))\,dt
=\int_0^1 t \frac d{dt} \E(\rk(A(t)))\,dt
\\&
=\Bigsqpar{t\bigpar{\E\rk(A(t))}}_{t=0}^1
 -\int_0^1 \E(\rk(A(t)))\,dt
\\&
= r(A) -\int_0^1 \E(\rk(A(t)))\,dt
\\&
=\int_0^1\big(\rk(A)-\E\rk(A(t))\big) \, dt
=\E\cstar.
  \end{split}
\raisetag{1.2\baselineskip}
\end{equation}
\end{proof}

\section{Conditional expectations and distributions} 
\label{Scond2}

In this section
we prove \refL{Lmom5} and \refT{tmom5a},
then present two corollaries and prove \refT{Thcond}.

\begin{proof}[Proof of \refL{Lmom5}]
We first prove \eqref{lmom5}. 
We assume for simplicity that 
all $2^{|A|}$ sums of costs over different sets of elements
are distinct.
Note that this can always be achieved by adding small independent
random increments 
$\gD c(a)$ to the costs. 
(If $\Sstar$ is a minimum structure, we may for example take
$\gD c(a)\sim U(0,\eps/|\Sstar|)$ for $a\in\Sstar$ and 
$\gD c(a)\sim U(\eps,2\eps)$ for $a\notin\Sstar$, for a small $\eps>0$,
so that $\Sstar$ is still a minimum structure.)
Since both sides of the equality in 
\eqref{lmom5} are continuous functions of the item costs,
the general case follows by continuity.

Define $\BB(F,E)$ to be a minimum-cost basis in $\IFE$,
so that
\begin{equation}\label{cIFE}
  c(\BB(F,E)) = \cstar(\IFE)+c(F\setminus E).
\end{equation}
We will only use this definition where $F$ is an independent set and
$E \subseteq F$, 
and we think of $\BB(F,E)$ as a cheapest basis 
forced to avoid $E$ but to include the rest of $F$.
For convenience we write $\BB(F,\set a)$ as $\BB(F,a)$.
Note that $\BB(F,E)$ has finite cost if and only if 
there exists a basis avoiding $E$,  \ie{} if $r(A\setminus E)=r(A)$.
By our assumption above, $\BB(F,E)$ is unique when it has finite cost;
in particular, $\Sstar=\BB(\emptyset,\emptyset)$ is unique. 

We claim that 
\begin{align}
F \subseteq \Sstar
\implies
(\forall a \in F) (\exists b \in A) \colon
\BB(\set a,a)=
\BB(F,a) = (\Sstar \setminus a) \cup \set{b} .
 \label{obs2}
\end{align}
Informally, this says that if $F$ is optimal, 
then any single-element exclusion
perturbs its best completion by just one further element.
To see this, let $a\in F\subseteq\Sstar$ 
and consider constructing $B=\BB(\set a,a)$ by the greedy algorithm.
(Since the matroid is bridgeless, $A\setminus\set a$ includes a basis,
so $\BB(\set a,a)$ has finite cost.)
This goes identically to the greedy construction of $\Sstar$
until element $a$ is reached:
$a$ is added to $\Sstar$ but not to $B$.
Then, 
every element accepted to $\Sstar$ is accepted to $B$,
but eventually some $b$ not chosen for $\Sstar$ is chosen for $B$;
denote by $\Sstar_i$ and $B_i$ the two collections of elements at this point.
Then $\Sstar_i \cup \set b = B_i \cup\set a$ is dependent, and
\begin{equation}
  r(B_i)=|B_i|=|\Sstar_i|=r(\Sstar_i)=r(\Sstar_i\cup\set b)
=r(B_i\cup\set a).
\end{equation}
Thus $r(\Sstar_i)=r(\Sstar_i\cup\set b)$
and $r(B_i)=r(B_i\cup\set a)$, and it follows by the submodular property
of the rank that for any subset $H\subseteq A$,
\begin{equation}
  r(B_i\cup H)
=r(B_i\cup\set a \cup H)
=r(\Sstar_i\cup\set b \cup H)
=r(\Sstar_i \cup H).
 \label{withH}
\end{equation}
Since the greedy algorithm selects the elements that increase the rank of
the set selected so far, it follows from \eqref{withH}
that from now on, exactly the same
elements will be selected for $B=B(\set a,a)$ and $\Sstar$.
Thus $\BB(\set a,a)=(\Sstar \setminus a) \cup \set b$.
The minimization for $\BB(\set a,a)$ 
is a relaxation of that for $\BB(F,a)$,
yet we have shown that 
$\BB(\set a,a) \cap F 
 =[(\Sstar \setminus a) \cup \set b] \cap F
 = F \setminus a$,
therefore $\BB(\set a,a) = \BB(F,a)$, which completes the proof of
\eqref{obs2}. 

Next, recalling 
\eqref{cIFE}, 
\begin{align}
\cstar(\Iz a)=c(\BB(\set a,\emptyset))-c(a)
&&\text{and}&&&
\cstar(\Iinf a)=c(\BB(\set{a},a)).
 \label{Ba0}
\end{align}
Furthermore, if $a\in F\subseteq\Sstar$, then
$B(\set a,\emptyset)=B(F,\emptyset)=\Sstar$, and
\eqref{obs2} shows that 
$B(\set a,a)=B(F,a)$. 
Thus by the VCG threshold definition \eqref{nota}, \eqref{Ba0}, 
\eqref{cIFE} again, and \eqref{vcgF},
\begin{align*}
\cVCG(a) 
&= \cstar(\Iinf a) - \cstar(\Iz a)
 = c(\BB(\set a,a)) - [c(\BB(\set a,\emptyset))-c(a)]
 \notag
\\&  = c(\BB(F,a)) - [c(\BB(F,\emptyset))-c(a)]
 \notag
 \\& = 
\cstar(I^{0,\infty}_{F\setminus\set a,a})
- \cstar(I^{0}_{F})
\notag \\ &= \vv(F,a)
 .
\end{align*}
This proves \eqref{lmom5}. 

We next prove \eqref{obs1x}, 
under the hypothesis that 
all $2^{|A|}$ sums of costs over different sets of elements
are distinct. 
(This hypothesis is natural and convenient, but could be weakened.) 

We first claim that 
\begin{align}
 F \subseteq \Sstar
 & \iff 
(\forall a \in F) \colon c(\BB(F,a)) \geq c(\BB(F,\emptyset)) .
 \label{obs1}
\end{align}
The forward implication is trivial:
$\BB(F,a)$ is a basis, $\Sstar$ is a minimum-cost basis,
and $\BB(F,\emptyset)$ is a basis minimized over a set
of possibilities including $\Sstar$.
The backward implication, through the contrapositive,
says informally that if $F$ is not contained within 
the minimum-cost basis
then excluding some single element $a$ improves it,
with $c(\BB(F,a))<c(\BB(F,\emptyset))$.
To prove this,
note that by assumption, all element costs are distinct.
Sort both $B=\BB(F,\emptyset)$ and $\Sstar$ in order of increasing cost
with elements $b_1,b_2,\ldots$ and $s_1,s_2,\ldots$ 
and write $B_i := \set{b_1,\ldots,b_i}$ and $\Sstar_i := \set{s_1,\ldots,s_i}$.
Let $i$ be the first index for which 
$s_i \neq b_i$.
Since $\Sstar$ can be constructed by the greedy algorithm
it must be that $c(s_i) \leq c(b_i)$,
and since all element costs are distinct, $c(s_i) < c(b_i)$.
Adding $s_i$ to $B$ creates a circuit $C_1$,
at least one of whose elements $b_j$ must have index $j \geq i$ 
(or else 
$C_1 \subseteq B_{i-1} \cup \set{s_i} 
 = \Sstar_{i-1} \cup \set{s_i} = \Sstar_i$,
contradicting independence of $\Sstar$).
Thus $c(s_i) < c(b_i) \leq c(b_j)$,
and $B' = B \cup \set {s_i} \setminus b_j$ 
has lower cost than $B$.
$B'$ is also a basis.
(If not, it has a circuit $C_2 \ni s_i$,
thus $B \cup \set {s_i}$ has 
circuits $C_1$ containing $s_i$ and $b_j$,
and $C_2$ containing $s_i$ but not $b_j$,
and by the 
circuit exchange axiom has a circuit $C_3$ with 
$C_3 \subseteq C_1 \cup C_2 \setminus s_i \subseteq B$,
contradicting independence of $B$.)
It cannot be that $b_j \in B \setminus F$, 
for then $B'\supseteq F$ which would contradict minimality of
$B=\BB(F,\emptyset)$. 
Thus $b_j \in F$ and we have shown that 
$c(\BB(F,b_j)) \leq  c(B') < c(\BB(F,\emptyset))$,
which completes the proof of \eqref{obs1}. 

By \eqref{cIFE} and definition \eqref{vcgF},
observation \eqref{obs1} is  equivalent to
\eqref{obs1x}.

We have proved \eqref{obs1x} 
under the hypothesis that 
all sums of costs over different sets of elements
are distinct. This immediately implies
\eqref{obs1y} and \eqref{obs1z} under the same hypothesis.
Moreover, the continuity argument at the beginning of the proof applies to 
\eqref{obs1y} and \eqref{obs1z}, which shows that they hold
for arbitrary costs,
concluding the proof of \refL{Lmom5}. 
(Note that the continuity argument does not apply to \eqref{obs1x}.)
\end{proof}

\begin{proof}[Proof of \refT{tmom5a}]
Fix $F$ and condition on the costs $c\vert_{A\setminus F}$.
(I.e., sample these costs first.)
As noted after \eqref{vcgF},
for $a \in F$
this determines 
$\vv(a)=\cstar(\Ixx{F\setminus{a}}{a})-\cstar(I^0_{F})$.

By hypothesis, the costs
$c(a)$ are random with independent continuous distributions, 
and thus almost surely 
$c(a)\neq \vv(a)$ for every $a\in F$.
This remains true when we condition also on $F\subseteq\Sstar$
since this is an event of positive probability.
By \eqref{obs1y}--\eqref{obs1z}, and $c(a)\neq\vv(a)$ a.s.\ as just observed,
this is the same as
further conditioning on $c(a)\le \vv(F,a)$ for all $a\in F$, and thus the
conditioned distributions of $c(a)$, $a\in F$, are independent and 
given by separately conditioning each $c(a)$ on $c(a)\le \vv(F,a)$. 
\end{proof}

\begin{remark}
Since the resulting conditional distribution in \refT{tmom5a}
depends only on the event
$F\subseteq\Sstar$ and the VCG costs $\cVCG(a)$, we obtain 
(by the definition of conditional expectation)
the same result by conditioning on these only.
\end{remark}

In the special case of \iid{} uniform costs, \refT{tmom5a} can be stated as
follows. 
\begin{corollary} \label{cmom4ii}
Consider a bridgeless 
matroid with ground set $A$ and
\iid{} costs $c(a)\sim U(0,1)$,
and let $F\subseteq A$ be an independent set.
Then, conditioning 
on the event $F\subseteq\Sstar$ and on the costs $c(a)$, $a\notin F$,
the VCG costs
$\cVCG(a)$, $a\in F$, are determined by \eqref{lmom5} and
the conditional distributions of $c(a)$,
$a\in F$, 
are independent with  $c(a) \sim U(0,\cVCG(a))$.

The same holds if we instead condition on $F\subseteq\Sstar$ and the values
$\cVCG(a)$, $a\in F$. 
\end{corollary}

\begin{proof}
\refT{tmom5a} shows that the conditional
distribution of $c(a)$ is uniform with
$c(a)\sim U(0,\min(\cVCG(a),1))$. 
Furthermore, by \eqref{le1}, or as a consequence of \eqref{obs2}, we 
have $\cVCG(a)\le1$. 
Hence, the conditional distribution is
$U(0,\cVCG(a))$.
\end{proof}

In the case when $F=\set a$ is a
singleton, \refC{cmom4ii} says that for any item $a$, 
conditioned on $a\in\Sstar$ and $\cVCG(a)$, the distribution of $c(a)$ is
uniform on $(0,\cVCG(a))$. This was seen already in the proof of Theorems
\ref{Th1} and \ref{Th2} in \refS{Spf1}; \refC{cmom4ii} is thus a
considerable generalization of this fact to several costs. 

The other extreme case is the case when $F$ is a basis.
In this case, $F\subseteq\Sstar$ is equivalent to $F=\Sstar$, and 
\refC{cmom4ii} reduces to the following.

\begin{corollary} \label{cmom2}
Consider a bridgeless  
matroid with costs $c(a)$ \iid{} uniform random variables, 
$c(a) \sim U(0,1)$.
Conditioned on a minimum structure $\Sstar$ and VCG payments
$\cVCG(a)$, $a\in\Sstar$,  
for its elements,
the costs $c(a)$, $a\in\Sstar$, of the elements in the minimum structure
are independent with  $c(a) \sim U(0,\cVCG(a))$.
\qed
\end{corollary}

\begin{proof}[Proof of \refT{Thcond}]
It follows immediately from \refC{cmom2} that conditioned on the minimum
structure $\Sstar$ and on the VCG payments $\cVCG(a)$ for all $a\in\Sstar$,
the conditional expectation of $\cstar=\sum_{a\in \Sstar} c(a)$
equals $\frac12\sum_{a\in \Sstar} \cVCG(a)=\frac12\cVCG$.
Hence the same holds also if we condition only on $\cVCG$.
\end{proof}

\begin{remark}
In  \refC{cmom2},
conditioning on the individual VCG values $\cVCG(a)$,
we obtain the conditional \emph{distribution} of the costs $c(a)$, $a\in \Sstar$,
and thus that of 
their sum $\cstar$.
In \refT{Thcond}, conditioning on the total VCG cost $\cVCG$ but not
on the individual values
$\cVCG(a)$, we obtain only a formula for the conditional
\emph{expectation} of $\cstar$; there is no general formula for the
conditional {distribution} of $\cstar$ given $\cVCG$, 
as is seen in \refE{EU}.
\end{remark}

\section{Variances and covariances}\label{Svar}

In this section, we continue to consider a bridgeless
matroid with costs $c(a)$ \iid{} uniform random variables, 
$c(a) \sim U(0,1)$,
and will prove \refT{Thvar}.
\ignore{
We begin by noting the special case $m=1$ of \refC{Cmom}:
\begin{align}\label{t2}
\E\bigpar{\cstar\cVCG}=\tfrac12\E\bigpar{(\cVCG)^2}.
\end{align}  
Equation \eqref{cov} follows from  \eqref{t2} and \refT{Th2}.
}

First, we use \refC{cmom2} to obtain 
a formula for
$\Var(\cstar)$.

\begin{theorem}\label{T22}
For a bridgeless
matroid with costs $c(a)$ \iid{} uniform random variables, 
$c(a) \sim U(0,1)$,
\begin{align}
\label{t22b}
\Var\bigpar{\cstar}&=\frac14\Var\bigpar{\cVCG}+
\frac14 \, \E\Bigpar{\sum_{a\in\Sstar}c(a)^2} 
\intertext{and}
\E\bigpar{(\cstar)^2}&=\frac14 \, \E\bigpar{(\cVCG)^2}+
\frac14 \, \E\Bigpar{\sum_{a\in\Sstar}c(a)^2}
\label{t22}
.
\end{align}  
\end{theorem}

\begin{proof}
Let $\cG$ denote the $\gs$-field 
generated by $\Sstar$ and the VCG costs 
$\cVCG(a)$, $a\in\Sstar$. 
(More formally, $\cG$ is
generated by the random variables 
$\set{ \one{a\in\Sstar}, \one{a\in\Sstar}\cVCG(a) \colon a\in A}$.
Recall that $\Sstar$ is a.s.\ unique.)
By \refC{cmom2}, conditioned on $\cG$, the random variables $c(a)$ for
$a\in\Sstar$ are independent with $c(a)\sim U(0,\cVCG(a))$.
Hence, by the law of total variance, recalling $\cstar=\sum_{a\in\Sstar} c(a)$,
\begin{align}
\Varb{\cstar}
 &=
 \Varb{\E (\cstar \mid \cG)} + \Eb{\Var( \cstar \mid \cG)}
 \notag \\&=
 \Var\Bigpar{ \sum_{a \in \Sstar}\E(c(a)\mid\cG)} 
   + \E\Bigpar{\sum_{a \in \Sstar} \Varb{c(a)\mid \cG)}}
 \notag \\&=
 \Var\Bigpar{\sum_{a \in \Sstar}\tfrac12 \cVCG(a)} 
   + \E\Bigpar{\sum_{a \in \Sstar} \tfrac1{12}(\cVCG(a))^2}
 \notag \\&=
 \Var\Bigpar{\frac12  \, \cVCG} 
   + \frac1{12} \, \E\Bigpar{\sum_{a \in \Sstar}(\cVCG(a))^2}.
\label{tv1}
\end{align}
The same conditioning gives
\begin{align}\label{tv2}
\E\Bigpar{\sum_{a\in\Sstar}c(a)^2} 
&=
\E\Bigpar{\E\Bigpar{\sum_{a\in\Sstar}c(a)^2\mid\cG}}
= \frac13 \, \E\Bigpar{\sum_{a\in\Sstar}(\cVCG(a))^2},
\end{align}
and \eqref{t22b} follows from \eqref{tv1} and \eqref{tv2}.

Finally, \eqref{t22} follows from \eqref{t22b} using \refT{Th2}. 
\end{proof}

\begin{remark}
  An alternative (but related) proof is obtained by writing 
$\cstar=\sum_{a\in A} \one{a\in\Sstar}c(a)$ and computing mixed second
 moments $\E\bigpar{\one{a,b\in F}c(a)c(b)}$
of the terms in this sum by conditioning and using
  \refC{cmom4ii} with $F=\set{a,b}$. This method can also be used to obtain
similar  (but more complicated) 
formulas for $\E(\cstar)^3$ and higher moments.
\end{remark}

We use also some of the ideas in the second proof of \refT{Th2}
(\refS{Spf2}), to establish
the following formula related to \refL{L1}.

\begin{lemma}\label{L2}
For a bridgeless
matroid with costs $c(a)$ \iid{} uniform random variables, 
\begin{align*}
\E\Bigpar{\sum_{a\in\Sstar}c(a)^2} &= 
\int_0^1 \bigpar{\rk(A)-\E\rk(A(t))}\, 2t\, dt  
=
\int_0^1 \E(\b(A(t))) \,t\, dt.
\end{align*}
\end{lemma}

\begin{proof}
$\Sstar$ is \as{} unique, and is then given by the greedy algorithm.
In this case we have, arguing as for \eqref{cseq},
\begin{equation}
  \begin{split}
\sum_{a\in \Sstar}c(a)^2
&=\sum_{a\in \Sstar} \int_0^\infty \one{t<c(a)}\,2t\, dt	
=\int_0^\infty \bigpar{|\Sstar|-|\Sstar\cap A(t)|}\,2t\, dt	
\\&
=\int_0^\infty \bigpar{r(A)-r(A(t)}\, 2t\,dt	.
  \end{split}
\raisetag{1.2\baselineskip}
 \label{ca2}
\end{equation}
(Note that \eqref{ca2} 
holds for a matroid  with arbitrary costs $c(a)$, 
even if there is not a unique minimum structure,
as long as $\Sstar$ is chosen by the greedy algorithm.)

The first equality in the lemma follows by taking the expectation, and the
second follows by an integration
by parts, like that in \eqref{ip}.  
\end{proof}

\begin{proof}[Proof of \refT{Thvar}]
The theorem consists of equations \eqref{cov}--\eqref{v2}.
We begin by noting the special case $m=1$ of \refC{Cmom}:
\begin{align}\label{t2}
\E\bigpar{\cstar\cVCG}=\tfrac12\E\bigpar{(\cVCG)^2}.
\end{align}  
Equation \eqref{cov} follows from  \eqref{t2} and \refT{Th2}.

By \eqref{cov}, $\Cov(2\cstar-\cVCG,\cVCG)=0$, and thus
\begin{equation}\label{stieg}
\Var(2\cstar) = \Var(\cVCG)+\Var(2\cstar-\cVCG),
\end{equation}
which yields \eqref{vvcg}.

Next, \eqref{vcstar} follows directly from \refL{L1}'s equation \eqref{m1}.
(The arithmetic is simplified using that $\Var(X)=\Cov(X,X)$
and that $\Cov(X,Y)$ is a bilinear form.)
This does not use our assumptions on the distribution of the costs,
so \eqref{vcstar} holds for any cost distribution, except that in general
one has to integrate to $\infty$.

Finally, \eqref{stieg} and \refT{T22}'s equation \eqref{t22b} yield
\begin{equation}
  \Var\bigpar{\cVCG-2\cstar}
=
4\Var\bigpar{\cstar}-\Var\bigpar{\cVCG}
=\E\Bigpar{\sum_{a\in\Sstar}c(a)^2} ,
\end{equation}
and thus \eqref{v1}--\eqref{v2} follow by \refL{L2}.
\end{proof}

We also want to point out an alternative proof of \eqref{v1}--\eqref{v2}
using martingale theory and the ideas in \refS{Spf2}.

\begin{proof}[Second proof of \eqref{v1}--\eqref{v2}]
The argument in the proof of \refL{Ldiff} really shows that
\begin{equation}\label{martin}
  \mm{x}:=
r(A(1-x)) -r(A) + \int_{1-x}^1 \b(A(t))\,\frac{dt}{t},
\qquad 0\le x\le 1,
\end{equation}
is a (continuous-time) martingale on $[0,1]$.
(With respect to the $\sigma$-fields
$\cF_x$ generated by $\set{c(a)\vee(1- x)}$, \ie, by the item arrivals after
time $1-x$. We may modify $\mmx$ to be right-continuous to conform with
standard theory; this makes no difference below and will be ignored.)
Note that $\mm 0=0$, and thus
$\E\mm{x}=0$, which by \eqref{martin} leads to a (perhaps more
formal) version of the proof of \refL{Ldiff}.

By \eqref{martin}, the martingale $\mm{x}$ has finite variation, with jumps
$\Delta \mm{x}=\Delta r(A(1-x))$ that are (\as) all $-1$.
Hence the quadratic variation $\qm{x}$ of the martingale is given by
\begin{equation}
  \begin{split}
  \qm{x}
&=\sum_{0\le y\le x} (\Delta\mm{y})^2
=\sum_{0\le y\le x} (-\Delta r(A(1-y)))
\\&
=r(A)-r(A(1-x)),	
  \end{split}
\end{equation}
and thus
\begin{equation}
\begin{split}
  \Var(\mm x)=\E \qm{x} = r(A)-\E r(A(1-x)).
\end{split}  
\end{equation}
(See e.g.\ \cite[Section II.6, in particular Corollary~3]{Protter} or
\cite[Theorem~26.6]{Kallenberg}.) 
Furthermore, since $\mmx$ is a martingale,
if $0\le x\le y\le 1$ then
\begin{equation}
  \Cov(\mm x,\mm y)
=
  \Cov(\mm x,\mm x) 
 = r(A)-\E r(A(1-x)) .
\end{equation}
Thus, integrating over $x,y\in\oi$,
\begin{equation}\label{vm}
  \begin{split}
\Var\Bigpar{&\int_0^1\mm x\,dx}	
=
\int_0^1\int_0^1\Cov\bigpar{\mm x,\mm y}\,dx\,dy
\\&=
2\int_0^1 \int_x^1\bigpar{r(A)-\E r(A(1-x))}\,dy\,dx
\\&=
2\int_0^1 (1-x)\bigpar{r(A)-\E r(A(1-x))}\,dx.
  \end{split}
\end{equation}

By \eqref{martin}, 
\begin{equation*}
  \begin{split}
\int_0^1  \mm{x}\,dx
&=
\int_0^1 \bigpar{r(A(1-x)) - r(A)} \,dx
+ \int_0^1\int_{1-x}^1 \b(A(t))\,\frac{dt}{t}\,dx
\\&=
\int_0^1 \bigpar{r(A(t)) - r(A)} \,dt
+ \int_0^1 \b(A(t))\int_{1-t}^1\,dx\,\frac{dt}{t}
\\&=
\int_0^1 \bigpar{r(A(t)) - r(A)} \,dt
+ \int_0^1 \b(A(t))\,dt
\\&=-\cstar+ (\cVCG-\cstar)
  \end{split}
\end{equation*}
by \eqref{m1}--\eqref{m2}.
Hence $\Var\bigpar{\cVCG-2\cstar}=\Var\bigpar{\int_0^1\mm x\,dx}$, 
and \eqref{v1} follows
from \eqref{vm}. 
As said earlier, \eqref{v2} follows from \eqref{v1} by integration
by parts.
\end{proof}

\section{Non-uniform cost distributions}\label{Snonuniform}

We have so far, with a few exceptions (notably \refT{tmom5a}), 
assumed that the item costs are
independent and uniformly distributed $U(0,1)$.
However, we can also derive some related results for random costs with 
other distributions
(not even necessarily identical), still assuming that the costs of different
items are independent.

For the next theorem, we say that a positive random variable 
has a \emph{(weakly) decreasing density function}
if it is absolutely continuous with a density function $f(x)$ on
$(0,\infty)$ satisfying $f(x)\ge f(y)$ when $0<x<y<\infty$,
and
a \emph{strictly decreasing density function}
if furthermore there exists $B\le\infty$ such that $f(x)>f(y)$ when
$0<x<y<B$ and $f(x)=0$ for $x>B$.

\refT{Th1} extends to this setting.
\begin{theorem} \label{Tmono1}
In the general VCG setting,
if the costs $c(a)$ are independent random
variables with decreasing density functions,
then
\begin{align}
 \E(\cstar)\le  \tfrac12\E(\cVCG) . 
\label{tmono1}
\end{align}
If furthermore every $c(a)$ has a strictly decreasing density function,
then the inequality \eqref{tmono1} is strict, 
provided $\E(\cstar)<\infty$.
\end{theorem}

\begin{proof}
As long as the costs are independent,
the proof of \refT{Th1} in \refS{Spf1}
up to \eqref{termideal0} applies to any continuous distributions.
If we condition on the costs of all items except $a$, which determines
$\cVCG(a)$, and then condition further on the event $c(a)\le\cVCG(a)$, 
then the conditional distribution of $c(a)$ has a decreasing density function
with support in $[0,\cVCG(a)]$, and thus the conditional expectation of $c(a)$
is at most $\frac12\cVCG(a)$, with strict inequality if the density function
is strictly decreasing. 
It follows that the expectation \eqref{termideal0} is at most half the
expectation \eqref{termVCG0}, and the result follows by summing over $a$ as
in the proof of \refT{Th1}.
\end{proof}

Likewise, in the matroid case, the corresponding result for conditional
expectations \refT{Thcond} generalizes 
as the theorem below,
now with an inequality where \refT{Thcond} had equality.

\begin{theorem} \label{Tmono2}
For a bridgeless matroid with costs $c(a)$ that are independent random
variables with decreasing density functions,
\begin{align}
  \E\bigpar{\cstar\mid\cVCG} \le \tfrac12\cVCG. \label{tmono2}
\end{align}
If furthermore every $c(a)$ has a strictly decreasing density function,
then the inequality \eqref{tmono2} is strict.
\end{theorem}

\begin{proof}
Let $F$ be a basis in the matroid and condition on $\Sstar=F$, and on the
VCG costs $\cVCG(a)$, $a\in F$. 
(Recall from \refT{tmom5a} that the VCG costs depend only on the item 
costs outside of $F$.)
By \refT{tmom5a} and our hypothesis, the
conditional distribution of each $c(a)$ has a decreasing density function
with support in $[0,\cVCG(a)]$, and thus the conditional expectation of $c(a)$
is at most $\frac12\cVCG(a)$, with strict inequality if the density function
is strictly decreasing. The result follows by summing over $a$ and relaxing
the conditioning as in the proof of \refT{Thcond}.
\end{proof}

Note that while the unconditional inequality \eqref{tmono1}
holds in the general VCG setting,
the conditional inequality \eqref{tmono2} does not always hold in
the non-matroid case, even for $U(0,1)$ costs; 
this will be shown in \refE{Epath}.

One case where these theorems apply (with strict inequalities) is
when the costs have independent exponential distributions.
The exponential distribution is discussed further in \refE{EUexp}.
\begin{corollary} \label{Cmono}
In the general VCG setting, 
if the costs $c(a)$ are independent random
variables with exponential distributions
(possibly with different means),
then \eqref{tmono1} holds with strict inequality.

Moreover, for a bridgeless 
matroid with such costs, \eqref{tmono2} holds with strict inequality.
\qed
\end{corollary}

\begin{remark}[\cite{AHW}]
In the general VCG setting, suppose
that each item
  $a$ exists in infinitely many copies $a_1,a_2,\dots$, with costs
$c(a_1)<c(a_2)<\dots$ given by
the points of  a Poisson process on $(0,\infty)$ with constant intensity
$\gl(a)$. Only the cheapest copy $a_1$ may be selected, and its cost
has an exponential distribution, so the nominal cost is the same as with
exponentially distributed costs. However, the existence of further copies 
may reduce the VCG cost, 
since $\cVCG(a_1)\le c(a_2)$. If we
condition on the costs of all copies of all items, except for $a_1$, 
then $c(a_1)\sim U(0,c(a_2))$
and since $\cVCG(a_1)\le c(a_2)$, we obtain $\E\cstar = \frac12\E\cVCG$ as
in the first proof of \refT{Th2}; see also \refR{cap1}.
\end{remark}

As noted in the introduction (and in \cite{AHW}), simple expressions result 
when the item costs are \iid{} 
with Beta distribution $B(\ga,1)$,
i.e., with density function
$\ga x^{\ga-1}$ on $(0,1)$, for some $\ga>0$. 
(Note that $\ga=1$ is the uniform case $U(0,1)$, and that the densities are
decreasing for $\ga\le 1$ but not for $\ga>1$.)

\begin{theorem} \label{Th1a}
In the general VCG setting, if the costs $c(a)$ 
are independent $B(\ga,1)$ random variables, with $\ga>0$,
then
\begin{align}\label{talpha1}
  \E(\cstar) \le \frac{\ga}{\ga+1}\E(\cVCG) .
\end{align}
For a bridgeless
matroid with these item costs, equality holds and, moreover,
\begin{align}\label{talpha2}
  \E\bigpar{\cstar\mid\cVCG} = \frac{\ga}{\ga+1}\cVCG .
\end{align}
\end{theorem}

\begin{proof}
  We argue again as in the proofs of 
Theorems \ref{Tmono1} and \ref{Tmono2}.
In the present case, the conditional
distribution of $c(a)$ given $c(a)\le v$ has support on $[0,\min(v,1)]$,
where it has density proportional to $x^\ga$.  
It follows that 
\begin{equation}
\E\bigpar{c(a)\mid c(a)\le v}=\frac{\ga}{\ga+1}\min(v,1)  
\end{equation}
for every $v>0$, and the result follows as in the proofs above.   
\end{proof}

If \refT{Th1} casts doubt on the practicality of a VCG auction,
with its factor-of-two expected overpayment,
\refT{Th1a} may be taken more hopefully:
for Beta-distributed costs, with large $\ga$,
the expected overpayment is small.

\section{Examples}\label{SEx}

\begin{example}[Minimum spanning tree] \label{EMST2}
Our motivating example (see also \refE{EMST})
was the minimum-cost spanning tree in the complete graph $K_n$
with \iid{} $U(0,1)$ edge costs.
This object has received extensive study.
It was shown by Frieze \cite{Frieze} that 
its expected cost (in our language, the expected nominal cost) satisfies
\begin{align} \label{eFrieze}
\E\cstar \to \zeta(3) \qquad  \text{as }  n \to \infty ; 
\end{align}
see \cite{CFIJS} for a recent sharper result.

Furthermore, it was shown in \cite{Janson} that 
$n^{1/2}(\cstar-\E\cstar)\dto N(0,\gs^2)$, as \ntoo, where, by W\"astlund 
\cite{Wastlund} (see also \cite{Janson+}),
$\gss=6\zeta(4)-4\zeta(3)$. It can be verified by using estimates in the
proof in \cite{Janson} that the variance converges too, \ie,
\begin{equation}\label{vMSTstar}
  \Var(\cstar) \sim \frac{\gss}n=\frac{6\zeta(4)-4\zeta(3)}n   
=\frac{1.68571\dots}n.
\end{equation}

Our theorems now yield 
the expectation and variance of the VCG cost.
(The expectation \eqref{eMST} was found more directly by \cite{CFMS}.)
\end{example}
\begin{theorem}\label{TMST}
For the minimum spanning tree in a complete graph with \iid{} $U(0,1)$ edge
weigths,
\begin{align}
\E \cVCG &= 2\E\cstar \to 2\zeta(3),  \label{eMST}
\\
\Var\bigpar{\cVCG} \label{vMST}
&\sim \frac{24\zeta(4)-18\zeta(3)}n                 
=\frac{4.33873\dots}n.
\end{align}
\end{theorem}

\begin{proof}
  The expectation \eqref{eMST} follows directly from 
\refT{Th2} and 
\eqref{eFrieze}. 
For the variance,
by \eqref{v1} and \refR{Rkk},
\begin{align}
\Var\bigpar{\cVCG-2\cstar}
&=
\int_0^1 \bigpar{\rk(A)-\E\rk(A(t))}\, 2t\, dt  
\\&
=
  \int_0^1 2t\big(\E\kk(A(t))-1)\big) \, dt.
\end{align}
The latter integral can be estimated using minor modifications of
estimates in \cite{CFIJS} 
(the main terms come from counting tree components; we omit the details),
which yields
\begin{align}\label{vmst9}
\Var\bigpar{\cVCG-2\cstar}
&=
  \int_0^1 2t\big(\E\kk(A(t))-1)\big) \, dt
\sim \frac{2\zeta(3)}n.                               
\end{align}
Finally, \eqref{vMST} follows from \eqref{vMSTstar} and \eqref{vmst9} by
\eqref{vvcg}.
\end{proof}

\begin{example}[Uniform matroid] \label{EU}
As a very simple example, consider 
the uniform matroid $\mathsf U_{n,k}$: it has $n$ elements (items)
and every
subset with exactly $k$ elements is a basis. Our objective is thus simply
to buy any $k$ items, and obviously the minimum structure consists of the $k$
cheapest items.
This simple example can easily be analyzed directly, without the general
theory above. 
We assume $1\le k\le n-1$ so that $\mathsf U_{n,k}$ is bridgeless.

If we order the items as $a\subb 1,\dots,a\subb n$ in order of increasing
cost (for simplicity assuming that there are no ties), 
the minimum structure is  $\Sstar=\set{a\subb1,\dots,a\subb k}$. 
Thus
\begin{align}\label{eu1}
  \cstar &= \sum_{i=1}^k c(a\subb{i}). 
\end{align}
If
$a\in\Sstar$, then in the instance $\Iinf a$, we select $a\subb{k+1}$
instead of $a$;
thus the incentive payment \eqref{defcstar} is $c(a\subb{k+1})-c(a)$ and the
VCG payment for $a$ is, by \eqref{split-VCG}, 
$c(a\subb{k+1})$. 
Hence
\begin{align}\label{eu2}
  \cVCG &= k c(a\subb{k+1}).
\end{align}
(In fact, it is easily seen that $\vcgt Ia=c(a\subb{k+1})$ if $a\in\Sstar$ and
$\vcgt Ia=c(a\subb k)$ if $a\notin\Sstar$.)

If
the costs $c(a_i)=U_i$ are \iid{} $U(0,1)$, we thus
have
\begin{align}
  \cstar &= \sum_{i=1}^k U\subb{i}, \label{eustar}
\\
  \cVCG &= k U\subb{k+1}, \label{euvcg}
\end{align}
where $U\subb i$, the $i$th order statistic of $\set{U_j}_1^n$, 
is well known to have Beta
distribution $B(i,n+1-i)$ 
and mean $\E U\subb i=i/(n+1)$.
Thus
\begin{align}
  \E \cstar &= \sum_{i=1}^k\frac{i}{n+1} = \frac{k(k+1)}{2(n+1)},
\\
\E\cVCG &= \frac{k(k+1)}{n+1},
\end{align}
in accordance with \refT{Th2}.

In this simple case, we can study 
the conditional distribution 
of $\cstar$
given $\cVCG$ in detail. 
It is well known that, conditioned on $U\subb{k+1}$, the 
distribution of $(U\subb i)_{i=1}^k$ equals the distribution of the order
statistics of $k$ independent uniform random variables on the interval
$(0,U\subb{k+1})$. 
Hence,
\begin{equation*}
\bigpar{(U\subb i)_{i=1}^k\mid U\subb{k+1}}  
\eqd (V\subb i U\subb{k+1})_{i=1}^k,
\end{equation*}
where $(V_i)_1^k$ also are \iid{} $U(0,1)$ random variables, independent of
$(U_j)_1^n$. Consequently, \eqref{eustar}--\eqref{euvcg} yield
\begin{equation}
\bigpar{\cstar\mid\cVCG}
\eqd S_k U\subb{k+1} 
=\frac{S_k}k\cVCG,
\end{equation}
where $S_k:=\sum_1^k V_i$ is independent of $\cVCG$.
Since $\E(S_k/k)=\frac12$, this is in accordance with \refT{Thcond}.
It also shows that the conditional distribution of $\cstar$ given $\cVCG$
may vary (although its expectation is always $\frac12\cVCG$). For example,
if $k=1$ then $\cstar$ is uniform on $(0,\cVCG)$, while if $k$ is large,
$\cstar$ is concentrated close to $\frac12\cVCG$.
In particular, $\E(\cstar)^2$ is not determined by $\E(\cVCG)^2$, and there
is no analogue of \refT{Th2} for the second moment.
Likewise there is no analogue of \refC{Cmom} for higher moments of $\cstar$
without $\cstar$ appearing on the right-hand side;
this is clear for $m=0$, and counterexamples can be generated for larger
values of $m$.

The variances are also easily computed, 
using that for $1 \leq i \leq j \leq n$, 
$\E(U\subb i \mid U\subb j)=\frac i j U\subb j$. 
We find that, for $1\le i\le j\le n$,
\begin{align}
  \E (U\subb iU\subb j)& = \frac{i(j+1)}{(n+1)(n+2)}
\\
  \Cov (U\subb i,U\subb j)& = \frac{i(n+1-j)}{(n+1)^2(n+2)},
\end{align}
and straightforward calculations from \eqref{eustar}, \eqref{euvcg} yield
\begin{align}
  \Var(\cstar) 
&= \frac{k(k+1)(4nk+2n+k+2-3k^2)}{12(n+1)^2(n+2)},
\\
  \Var(\cVCG) 
&= \frac{k^2(k+1)(n-k)}{(n+1)^2(n+2)}.
\end{align}

From \eqref{cov} this allows us to obtain
$\Cov(\cstar,\cVCG)$;
we can also calculate it too directly and verify equality.
By \eqref{vvcg}, $\Var(\cVCG-2\cstar)=k(k+1)(k+2)/3(n+1)(n+2)$, and
this can be checked against its calculation by
\eqref{v1}--\eqref{v2}.

In particular, if $n\to\infty$ and $1\ll k\ll n$, then 
$\Var(\cstar)/\Var(\cVCG)\to1/3$, and more precisely
the covariance matrix of
$(\cstar,\cVCG)$ is asymptotic to
\begin{equation}
  \frac{k^3}{n^2}
\begin{pmatrix}
\xfrac{1}3 & \xfrac{1}2 \\ 1/2 & 1
\end{pmatrix}  .
\end{equation}
If $n\to\infty$ with $k\ge1$ fixed, we have instead
$\Var(\cstar)/\Var(\cVCG)\to (2k+1)/(6k)$.

The latter limit demonstrates that $\Var(\cstar)$ and $\Var(\cVCG)$
are not in fixed proportion,
so there is not a direct analogue of \refT{Th2} for variance.
\end{example}

\begin{example}[Uniform matroid, exponential distribution] \label{EUexp}  
We continue to consider 
the uniform matroid $\mathsf U_{n,k}$ as in \refE{EU}, but now consider
costs $c(a)$ that are \iid{} and exponential $\Exp(1)$.
Let $X\subb i=c(a\subb i)$, the $i$th smallest cost.
It is well-known that the increments $X\subb i-X\subb{i-1}$ 
(where $X\subb0=0$)
are independent random variables with $X\subb i-X\subb{i-1}\sim
\Exp(1/(n-i-1))$. Consequently,
\begin{equation}
  \E X\subb i = \sum_{j=1}^i \frac{1}{n-j+1}
\end{equation}
and thus by \eqref{eu1}--\eqref{eu2},
\begin{align}
\E \cstar &= \sum_{j=1}^k \frac{k-j+1}{n-j+1}
= \sum_{j=1}^k \frac{j}{n-k+j},
\\
\E \cVCG &= k\sum_{j=1}^{k+1} \frac{1}{n-j+1}.
\end{align}
Not only are these not a factor of 2 apart, 
as we already knew from \refC{Cmono},
but there doesn't seem to be any simple general relation.

If for simplicity we consider the case $k=1$, where our aim just is to
select the cheapest item, it follows from \refT{tmom5a} that the conditional
distribution of $\cstar$ given $\cVCG=v$ has density $e^{-x}/(1-e^{-v})$ on
$(0,v)$. A simple calculation yields
\begin{align} \label{expRatio}
 \E\bigpar{\cstar\mid\cVCG=v}
=\frac{1-e^{-v}-ve^{-v}}{1-e^{-v}}
=\frac{e^{v}-1-v}{e^{v}-1}.
\end{align}
This is always $<v/2$, as shown by \refC{Cmono}, with asymptotic equality as
$v\to0$.
\end{example}

\begin{example}[Paths]\label{Epath}

\cite{CFMS} considered the minimum-\weight 
path between two given vertices
in $K_n$ with \iid{} 
$\Exp(1)$
costs of the edges. They showed that
$\E\cVCG\sim 2\E\cstar$ as $n\to\infty$. 
It is not hard to show that
the same holds with \iid{}
$U(0,1)$ costs.
In these cases, however, there is not equality in \eqref{th1}.
(This is not a matroid case, so \refT{Th2} does not apply.)

For a simple explicit example, consider $n=3$. Denoting the three edges in
$K_3$ by
$a_1,a_2,a_3$ and their costs by $X_1,X_2,X_3$,  
and let the two given vertices be those joined by $\set{a_1}$
and $\set{a_2,a_3}$, so that
\begin{equation}\label{ep1}
  \cstar = \min\set{X_1,X_2+X_3}.
\end{equation} 
It follows from this and 
\eqref{nota} that $\cVCG(a_1)=X_2+X_3$,  $\cVCG(a_2)=(X_1-X_3)_+$ and
$\cVCG(a_3)=(X_1-X_2)_+$, and thus
\begin{align}
  \cVCG=
  \begin{cases}
	X_2+X_3, &\text{if } X_1\le X_2+X_3, \\
	2X_1-X_2-X_3, &\text{if } X_1\ge X_2+X_3.
  \end{cases}
\end{align}
Thus,
\begin{equation}\label{ep3}
  \cVCG=\max(X_2+X_3,\,2X_1-X_2-X_3).
\end{equation}
Simple calculations, which we omit, show that 
in the $U(0,1)$ case,
$\E\cstar = 11/24$ while $\E\cVCG = 13/12$
(a check on these is that from \eqref{ep3},
$2\cstar+\cVCG=2X_1+X_2+X_3$, thus $2\E\cstar+\E\cVCG=2$), 
showing strict inequality in
\eqref{th1}.
Similarly, in the $\Exp(1)$ case, $\E\cstar = 3/4$ while $\E\cVCG = 5/2$,
with strict inequality in \eqref{tmono1} as shown in \refC{Cmono}.

Moreover, still for $n=3$ and with $U(0,1)$ costs, 
the joint density of $\cstar$ and $\cVCG$ 
in $[0,1]\times[0,2]$
can easily be calculated from \eqref{ep1}--\eqref{ep3} to be
\begin{equation}\label{epf}
f(x,y)=
  y\one{x\le y\le 1}
+ (2-y)\one{1<y\le 2}
+ \frac{x}2\one{x\le \min(y,2-y)}.
\end{equation}
For example, $f(x,1)=1+x/2$, $x\in[0,1]$, and thus the conditional
distribution of $(\cstar\mid\cVCG=1)$ has density
$(4+2x)/5$ and conditional expectation
$\E(\cstar\mid\cVCG=1)=8/15>1/2$, showing that the inequality
\eqref{tmono2}
does not necessarily hold in the non-matroid setting.
More generally, it follows from \eqref{epf} that the 
density of $\cVCG$ is
\begin{equation}
f_{2}(y)=
  \begin{cases}
\frac54 y^2 	&\text{if } 0\le y\le 1, \\
(2-y)+\tfrac14(2-y)^2
      &\text{if } 1\le y\le 2,
  \end{cases}
\end{equation}
and the conditional expectation   $\E(\cstar\mid\cVCG)$
is 
\begin{equation*}
  \E(\cstar\mid\cVCG=y)=
  \begin{cases}
\frac{2y^3/3}{f_{2}(y)}
=\frac{8}{15}y 	&\text{if } 0\le y\le 1, \\
\frac{\tfrac12(2-y)+\tfrac16(2-y)^3}{f_2(y)} =\frac{6+2(2-y)^2}{12+3(2-y)}
   &\text{if } 1\le y\le 2.
  \end{cases}
\end{equation*}
In this example, $\E(\cstar\mid\cVCG)$
is a \emph{decreasing} function of $\cVCG$ in the interval
$\cVCG\in[1,6-\sqrt{19}]$, which shows that the behaviour of the
conditional expectation can be quite complicated and unexpected in 
the general non-matroid case.
\end{example}

\end{document}